\documentclass[hidelinks, 12pt]{article}

\usepackage{subcaption}
\usepackage{caption}
\usepackage{subcaption}
\usepackage{fullpage}
\usepackage[utf8]{inputenc}
\usepackage{appendix}
\usepackage{graphicx}
\usepackage{amsmath}
\usepackage{amssymb}
\usepackage{amsthm, amsfonts}
\usepackage{booktabs}
\usepackage{listings}
\usepackage{hyperref}
\usepackage{float}
\newtheorem{theorem}{Theorem}[section]
\newtheorem{corollary}{Corollary}[theorem]
\newtheorem{lemma}[theorem]{Lemma}

\parskip = \baselineskip
\title{Constant Power Root Market Makers}

\author{Mike Wu, Will McTighe \\ \small\texttt{\{mike, will\}@paretolabs.xyz}}

\date{February 2022}

\begin{document}

\maketitle

\begin{abstract}
The paper introduces a new type of constant function market maker, the \textit{constant power root market marker}. We show that the constant sum (used by mStable), constant product (used by Uniswap and Balancer), constant reserve (HOLD-ing), and constant harmonic mean trading functions are special cases of the constant power root trading function. We derive the value function for liquidity providers, marginal price function, price impact function, impermanent loss function, and greeks for constant power root market markers. In particular, we find that as the power $q$ varies from the range of $-1$ to $1$, the power root function interpolates between the harmonic ($q=-1$), geometric ($q=0$), and arithmetic ($q=1$) means. This provides a toggle that trades off between price slippage for traders and impermanent loss for liquidity providers. As the power $q$ approaches 1, slippage is low and impermanent loss is high. As $q$ approaches to -1, price slippage increases and impermanent loss decreases. 
% We then propose an extension: the \textit{dynamic power root market maker}, which automatically adjusts the power in response to market activity, and show simulation results.
% We perform simulations with real world data to test theoretical claims made.
\end{abstract}

\section{Preface}

This paper is part of an exploration of alternative approaches to AMM curvature to incentivize trading and liquidity provision in specific use cases. The goal of an efficient market is to provide the most liquidity at the lowest price impact. In service of that goal, we had two initial hypotheses. First, there is low liquidity for long tail assets in existing automated market makers (AMMs) because it is unattractive for liquidity providers, absent liquidity incentives. We hypothesized that greater curvature than the constant product function could incentivize additional liquidity provision. Second, we reasoned that there could be an alternative AMM for stableswaps, providing more protection to LPs in edge cases, like the recent events of Terra. Flexing power could be a gas efficient approach to adjusting curvature but introduces issues such as reserve depletion and round-trip arbitrage. We explored these in follow up research. If you'd like to discuss these topics in more detail, please reach out.

\section{Constant Function Market Makers}

Constant function market makers (CFMMs) are a family of AMMs commonly used to parameterize asset exchange mechanisms on public blockchains. Critically, CFMMs rely on liquidity providers (LPs) who provide tokens to a supply of reserves for use by a smart contract. The contract contains code that executes swaps for traders. To preserve market stability, swaps must ensure a function of the reserves is constant. This function is known as the trading function, also called the scoring rule or the invariant.

For simplicity, we will assume the case of two tokens, although it should be straightforward to generalize to more than two. Let $\psi: \mathbf{R} \times \mathbf{R} \rightarrow \mathbf{R}$ be an invariant, and $x, y \in \mathbf{R}_+$ be reserves. A reserve represents the amount of a token available for swaps in the contract pool. Let $x', y' \in \mathbf{R}_+$ denote new values for the reserves after a trade. A trade is allowed if and only if $\psi(x', y') = k$ for some constant value $k \in \mathbf{R}_+$, hence the name \textit{constant function}.

Three examples of trading rules are, borrowing the notation from \cite{angeris2021replicating}:

\textbf{Constant reserve } We begin with a trivial invariant $\psi_{\text{rsv}}(x, y) = \left\{\begin{matrix}
0 & \textup{if } x = k_1 \textup{ and } y = k_2 \\
-\infty & \text{otherwise} \\
\end{matrix}\right.$ for $k_1, k_2 \in \mathbf{R}$, which only allows trading if it keeps the reserves at fixed values $a, b \in \mathbf{R}^+$. When $\psi$ evaluates to $-\infty$, trading is not permitted. This is also known as HODL-ing.

\textbf{Constant sum } Define $\psi_{\text{sum}}(x, y) = \left\{\begin{matrix}
0 & \textup{if } x + y = k \\
-\infty & \text{otherwise} \\
\end{matrix}\right.$, also called the ``arithmetic mean''. This function is only usable for asset pairs where the price ratio is fixed (e.g. USDT-USDC or wBTC-renBTC). Technically, it has no slippage (since price is fixed) and no impermanent loss (since we can always freely swap one $x$ for one $y$), but no price adjustments as reserves fluctuate, meaning a reserve can be depleted entirely. 
% However, the more intuitive interpretation is that as a trading function tends to $(x + y)$, becoming flatter, slippage is low and impermanent loss is very high. 
For a real world example, mStable is a DEX for stablecoins that uses a constant sum trading function.

\textbf{Constant product } Define $\psi_{\text{prod}}(x, y) = \left\{\begin{matrix}
0 & \textup{if } xy = k \\
-\infty & \text{otherwise} \\
\end{matrix}\right.$, also called a ``geometric mean''. We can see that $y$ must grow exponentially large as $x$ nears zero, discouraging any traders from fully depleting reserves. (Similarly, the property holds if $y$ is the token nearing zero). The unfortunate consequence of this design is susceptibility to slippage. Uniswap \cite{angeris2019analysis,adams2021uniswap} and Balancer \cite{martinelli2019non,evans2020liquidity} are instances of a constant product. 

\textbf{Curve \cite{egorov2019stableswap} } Define $\psi_{\text{curve}}(x, y) = \left\{\begin{matrix}
0 & \textup{if } \alpha(x+y) + \beta(xy) = k \\
-\infty & \text{otherwise} \\
\end{matrix}\right.$, an interpolation between the constant sum and product functions. As $\left(\frac{\alpha}{\beta}\right)$ increases, this converges to the constant sum. As $\left(\frac{\alpha}{\beta}\right)$ decreases, this converges to the constant product. In this paper, we explore another method to toggle between arithmetic and geometric means.

\textbf{Constant harmonic mean } Define $\psi_{\text{har}}(x, y) = \left\{\begin{matrix}
0 & \textup{if } \frac{1}{\left(\frac{1}{x} + \frac{1}{y}\right)} = k \\
-\infty & \text{otherwise} \\
\end{matrix}\right.$. Like the geometric mean, there is no reserve depletion. To the best of our knowledge, this function has not been implemented in practice yet. One of the contributions of this paper is to explore the properties of this trading function and similar relatives.

In this paper, we will propose a new family of constant function market makers using a power root. 
Special cases of these power root functions (in particular, those with non-negative power) have been previously explored by Clipper \cite{othman2021} and Yieldspace \cite{niemerg2020}. 

\subsection{Marginal Price and Impermanent Loss}

We review the marginal price and impermanent loss for the constant sum and product functions, all well known results. 

\begin{lemma}
The marginal price for the constant sum function is $M_{\textup{sum}} = 1$. The marginal price for constant product function is $M_{\textup{prod}} = \frac{y}{x}$.
\label{lem:price}
\end{lemma}
\begin{proof}
\small
Consider $(x - \Delta x) + (y + \Delta y) = k$. Solving for $\Delta y$, we obtain $\Delta y = k - x + \Delta x - y$. Then, $\lim_{\Delta x \rightarrow 0}\left(\frac{\Delta y}{\Delta x}\right) = \lim_{\Delta x \rightarrow 0}\left(\frac{k - x + \Delta x - y}{\Delta x}\right) = 1$. Next, consider, $(x - \Delta x)(y + \Delta y) = k$. Again, solving for $\Delta y$, we get $\Delta y = \frac{y\Delta x}{k - \Delta x}$. So $\lim_{\Delta x \rightarrow 0}\frac{\Delta y}{\Delta x} = \lim_{\Delta x \rightarrow 0}\frac{y}{x - \Delta x} = \frac{y}{x}$.
\end{proof}

In other words, the price ratio for the constant sum function is fixed. The price is independent of the token reserves. This is not true for the constant product function.

\begin{lemma}
Let $M' = \alpha M$ be a change in marginal price, $\alpha > 0$. Ignoring trading fees, the impermanent loss function for constant product function is $I_{\textup{prod}} = \frac{2\sqrt{\alpha}}{\alpha + 1} - 1$. For the constant sum function, $I_{\textup{sum}} = 0$ since the price marginal ratio of $x$ and $y$ must be fixed.
% WILL: need to explain I_sum = 0 because effective impermanent loss under these conditions are very high
\label{lem:og:imp}
\end{lemma}
\begin{proof}
\small
Let $x', y'$ be the new reserves corresponding to $M'$. Define $U(x, y, M) = Mx + y$, the portfolio value at price $M$. Then $I = \left(\frac{U(x', y', M') - U(x, y, M')}{U(x, y, M')}\right)$, the normalized amount the liquidity provider would have been better off by keeping tokens out of the pool. 

Starting with constant sum, we have $U(x', y', M') = M' x' + y' = x' + y' = k$, and $U(x, y, M') = M' x + y = x + y = k$. So $I_{\text{sum}} = \frac{k-k}{k} = 0$.

The constant product is slightly more intricate. Since $M = \frac{y}{x}$, we have $y = Mx$ and so $xy = x^2 M = k$. Thus, $x = \sqrt{\frac{k}{M}}$ and $y = \sqrt{kM}$. Now, $U(x', y', M') = M' \sqrt{\frac{k}{M'}} + \sqrt{kM'} = 2\sqrt{\alpha kM}$. Similarly, $U(x, y, M') = (\alpha M)\sqrt{\frac{k}{M}} + \sqrt{kM} = (\alpha+1)\sqrt{kM}$. So $I_{\text{prod}} = \frac{2\sqrt{\alpha}\sqrt{kM} - (\alpha+1)\sqrt{kM}}{(\alpha+1)\sqrt{kM}} = \frac{2\sqrt{\alpha}}{\alpha + 1} - 1$.
\end{proof}

\section{Value Functions for Liquidity Providers}

Given a trading function, a reasonable question to ask is what is the LP payoff function? Similarly, given a LP payoff function, what is the corresponding trading function?  
To answer these questions, prior work \cite{angeris2021replicating} formulates the trading and value functions as dual convex optimization problems. We summarize the findings below:

Let $a, b \in \mathbf{R}_+$ denote the true (external market) price of $x, y$, respectively. Define a value function $V: \mathbf{R}_+ \times \mathbf{R}_+ \rightarrow \mathbf{R}$ as a function that maps true prices to a payoff amount \cite{angeris2020improved}. Then, prior work \cite{angeris2020improved} has shown that 
\begin{equation}
    V(a, b) = \inf_{x, y \in \mathbf{R}_+} \{ax + by | \psi(x, y) = k \}
    \label{eq:forward}
\end{equation}
Equation~\ref{eq:forward} shows that traders and LPs are at odds in a zero-sum game. That is, traders achieve maximum profit when LPs achieve minimum value, preserving the constant function. (The converse is also true.) This intuition leads to the foundational result that 
\begin{equation}
    \psi(x, y) = \inf_{a, b \in \mathbf{R}} \{ax + by - V(a, b) \}.
    \label{eq:backward}
\end{equation}
In other words, the trading and value functions are Fenchel conjugates of one another. This duality does not hold for all value functions $V$. We require $V$ to be \textit{consistent}, which is composed of four requirements: concavity, non-decreasing, non-negative, and 1-homogeneous.

We list the value function for popular trading functions, and defer to \cite{angeris2021replicating} for derivations. It is straightforward to verify all value functions listed below are consistent.

\textbf{Constant reserve } $V_{\text{rsv}}(a, b) = a + b$, the sum value function.

\textbf{Constant sum } $V_{\text{sum}}(a, b) = \min \{a, b\}$, the minimum value function.

\textbf{Constant product } $V_{\text{prod}}(a, b) = \sqrt{ab}$, the product value function.

While enlightening, this analysis so far leaves us wanting \textit{more}. What is the relationship between these three value functions? Are there other value functions of the same family? Why does the constant product trading function map to a product value function but constant sum maps to a minimum function? Next, we take a step towards an answer.

\section{Constant Power Root Market Makers}

We structure the section as follows: first, we introduce a family of value functions that contain the constant reserve, sum, and product functions as special cases; second, we derive the corresponding trading function; third, we compute several properties of interest, such as marginal price and impermanent loss. Proofs are provided inline.

\begin{figure}[h!]
\centering
\includegraphics[width=\textwidth]{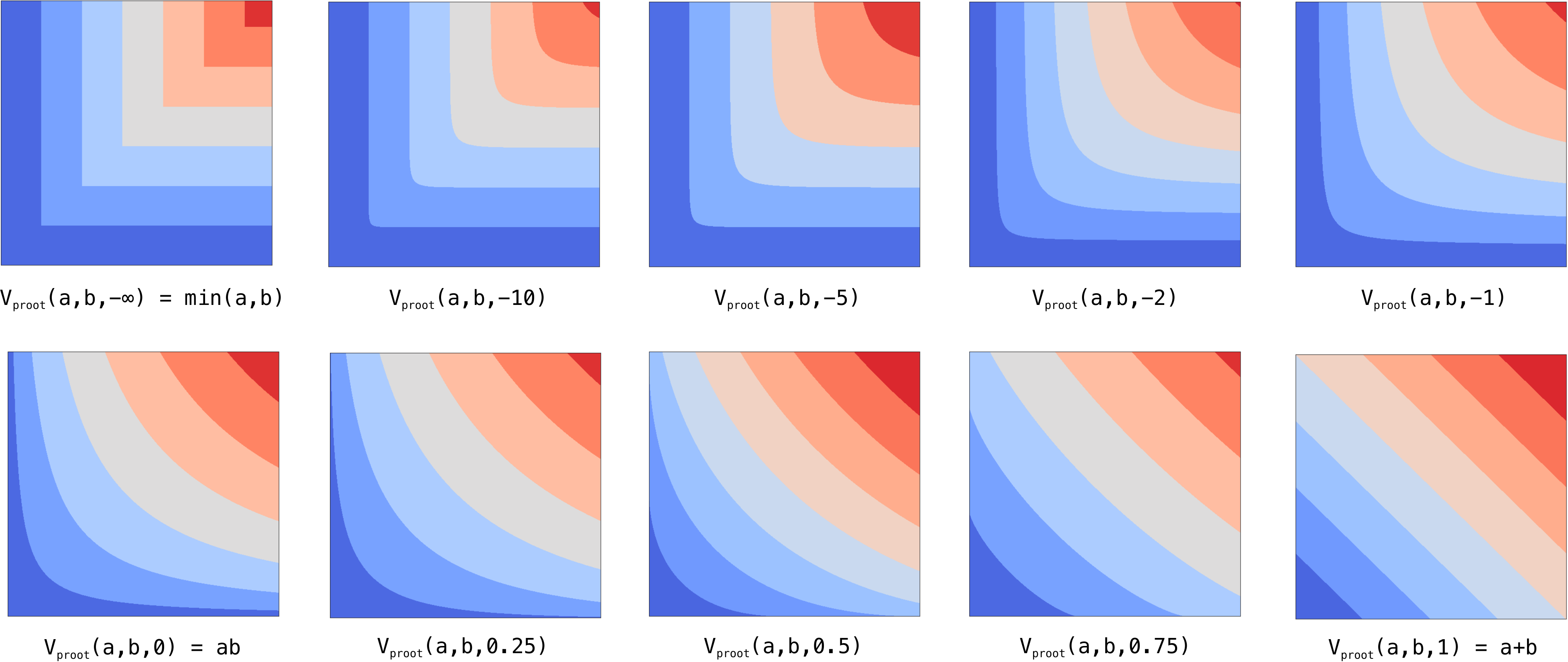}
\caption{Examples of power root value functions that interpolate between $p=0$ (constant product) and $p=1$ (constant reserve) and $p=-\infty$ (constant sum). }
\label{fig:invariant}
\end{figure}

\subsection{Power Root Value Function}

Define the power root value function as 
\begin{equation}
    V_{\text{pow}}(a, b, p) = (a^p + b^p)^{\frac{1}{p}}
    \label{eq:power_root}
\end{equation}
for a fixed constant $p\leq 1$. Note that $p$ does not have to be an integer. Equation~\ref{eq:power_root} is equivalent to a $p$-norm when $a, b, p \geq 0$. However, the two are not identical in properties, particularly in terms of concavity, as we will show in Section~\ref{sec:proot_trade}.

We claim that all the value functions corresponding to the constant reserve, sum, and product invariants are special cases of the power root value function. Consider:

\begin{lemma}
Suppose $a, b > 0$.
If $p = 1$, then $V_{\text{pow}}(a, b, 1) = V_{\text{rsv}}$. As $p \rightarrow 0$, we have $V_{\text{pow}}(a, b, 0) \rightarrow V_{\text{prod}}$. As $p \rightarrow -\infty$, we have $V_{\text{pow}}(a, b, -\infty) \rightarrow V_{\text{sum}}$.
\label{eq:special}
\end{lemma}
\begin{proof}
\small
Begin with $p = 1$. We have $V_{\text{pow}}(a, b, 1) = a + b = V_{\textup{rsv}}$. 

Next, consider $p \rightarrow -\infty$. Note that $(a^p + b^p)^{\frac{1}{p}} \geq \left(2 \min \{a, b \}^p\right)^{\frac{1}{p}} = 2^{\frac{1}{p}}\min\{a, b\}$ by definition of minimum. But also $(a^p + b^p)^{\frac{1}{p}} \leq \left(\min\{a, b\}^p\right)^{\frac{1}{p}} = \min \{a, b\}$ since $p < 0$. By squeeze theorem, we have $(a^p + b^p)^{\frac{1}{p}} = \min\{a, b\}$.

Finally, consider $p \rightarrow 0$. If $a, b \geq 0$, then: 
\begin{align*}
    \lim_{p \rightarrow 0} (a^p + b^p)^{\frac{1}{p}} &= \exp \log \lim_{p \rightarrow 0} (a^p + b^p)^{\frac{1}{p}} = \exp \lim_{p \rightarrow 0} \log (a^p + b^p)^{\frac{1}{p}}  \\
    &=  \exp \lim_{p \rightarrow 0} \frac{\log (a^p + b^p)}{p} = \exp \lim_{p \rightarrow 0} \frac{\left(\frac{a^p \log a + b^p \log b}{a^p + b^p}\right)}{1} \\
    &= \exp \left(\frac{\log a + \log b}{2}\right) = \sqrt{ab}
\end{align*}
\end{proof}

\textbf{Economic Interpretation.} We can view value functions as \textit{production functions} in economic theory. For instance, $V_{\text{prod}}$ is equivalent to the Cobb-Douglas production function, $V_{\text{rsv}}$ is equivalent to the perfect substitute (linear) production function, and $V_{\text{sum}}$ is equivalent to the Leontief production function. Further, $V_{\text{pow}}$ is equivalent to the constant elasticity of substitution, of which the perfect substitute, Leontief, and Cobb-Douglas are special cases of. This matches the result shown in Lemma~\ref{eq:special}. 

\subsection{Power Root Trading Function}
\label{sec:proot_trade}

Before we can apply Equation~\ref{eq:backward} to derive a trading function, we must verify that $V_{\text{pow}}$ satisfies the four conditions for consistency \cite{angeris2021replicating}.

\begin{lemma}
$f(a, b) = (a^p + b^p)^{\frac{1}{p}}$ for $p \leq 1$, $a, b \in \mathbf{R}$ is concave.
\label{lem:concave}
\end{lemma}
\begin{proof}
\small
We first show that $f$ is quasiconcave, then 1-homogenous, which implies concavity.

A function is quasiconcave if it is a monotonically increasing function of a concave function. 
Let $g(a, b) = a^p + b^p$. If $p > 1$, this is  not concave. Then, if $0 \leq p \leq 1$, set $f(a, b) = g(a, b)^{\frac{1}{p}}$. Clearly $f$ is a monotonic function of $g$. It remains to show that $g$ is concave. To see this, note that $\nabla_{p,a} g = pa^{p-1}$, and $\nabla^2_{p, a} g = p(p-1)a^{p-2} \leq 0$ since $p - 1 \leq 0$. Similarly, $\nabla^2_{p, b} g \leq 0$. Thus, the Hessian of $g$ is a diagonal matrix with non-positive entries, and hence semidefinite. We conclude $g$ is concave. Next, suppose $p < 0$. Then $f = g^{\frac{1}{p}}$ is a monotonically decreasing function of $g$. But then $f$ is a monotonically increasing function of $-g$. The Hessian of $-g$ has entries $-p(p-1)a^{p-2}$ which are again non-positive, making $-g$ concave.

$f$ is also 1-homogeneous. To see this, note $f(ka, kb) = \left((ka)^p + (kb)^p\right)^{\frac{1}{p}} = \left(k^p \left(a^p + b^p\right)\right)^{\frac{1}{p}} = k\left(a^p + b^p\right)^{\frac{1}{p}} = kf(a, b)$. Given quasiconcavity and homogeneity, we have concavity.
\end{proof}

Lemma~\ref{lem:concave} shows the power root is a concave function for $p \leq 1$.  

As a corollary, we obtain our first theorem.

\begin{theorem}
The value function $V_{\text{pow}}(a, b) = (a^p + b^p)^{\frac{1}{p}}$ is consistent.
\label{thm:consistent}
\end{theorem}
\begin{proof}
\small
Lemma~\ref{lem:concave} shows concavity and homogeneity. For $a,b \geq 0$, it is true that $V_{\text{pow}}(a, b) \geq 0$ for all $p \leq 1$ so $V_{\text{pow}}$ is non-negative. Finally, pick $a' \geq a$ and $b' \geq b$. Then $V_{\text{pow}}(a', b') = ((a')^p + (b')^p)^{\frac{1}{p}} \geq (a^p + b^p)^{\frac{1}{p}} = V_{\text{pow}}(a, b)$, showing that it is non-decreasing.
\end{proof}

Now that we have shown $V_{\text{pow}}$ is consistent we are free to compute $\psi_{\text{pow}} = \inf_{a, b \in \mathbf{R}} \{ax + by - V_{\text{pow}}(a, b) \}$ as in Equation~\ref{eq:backward}. We claim the following:

\begin{theorem}
$\psi_{\text{pow}}(x, y) = \left\{\begin{matrix}
0 & \text{if } \left(x^q + y^q\right)^{\frac{1}{q}} \leq 1 \\
-\infty & \text{otherwise} \\
\end{matrix}\right.$ where $q = \frac{p}{p-1}$.
\label{thm:poweroot}
\end{theorem}
\begin{proof}
\small
We will show this in the more general case of $n$ tokens $\mathbf{x} = (x_1, \ldots x_n)$ with prices $\mathbf{c} = (c_1, \ldots, c_n)$. By definition, we have $\psi(\mathbf{x}) = \inf_{\mathbf{c}}\left\{x^T c - \left(\sum_i c_i^p\right)^{\frac{1}{p}}\right\}$. Note that this is a convex objective as $x^T c$ is affine and the negation of a concave function convex. Then:
\begin{equation*}
    \psi(\mathbf{x}) = \inf_{\mathbf{c}}\left\{\mathbf{x}^T \mathbf{c} - \left(\sum_i c_i^p\right)^{\frac{1}{p}}\right\} = -\sup_{\mathbf{c}}\left\{-\mathbf{x}^T \mathbf{c} + \left(\sum_i c_i^p\right)^{\frac{1}{p}} \right\} = -f^*(-\mathbf{x})
\end{equation*}
where $f^*$ is the convex conjugate of $f = - \left(\sum_i c_i^p\right)^{\frac{1}{p}}$. This can be written as the closed form $f^*(-\mathbf{x}) = \left\{\begin{matrix}
0 & \text{if } \left(\sum_i x_i^q \right)^{\frac{1}{q}} \leq 1 \\
\infty & \text{otherwise} \\
\end{matrix}\right.$ where $q = \frac{p}{p-1}$. Then $\psi_{\text{pow}}(\mathbf{x}) = -f^*(-\mathbf{x}) =  \left\{\begin{matrix}
0 & \text{if } \left(\sum_i x_i^q \right)^{\frac{1}{q}} \leq 1 \\
-\infty & \text{otherwise} \\
\end{matrix}\right.$. For two tokens, $\psi_{\text{pow}}(x, y) =   \left\{\begin{matrix}
0 & \text{if } \left(x^q + y^q \right)^{\frac{1}{q}} \leq 1 \\
-\infty & \text{otherwise} \\
\end{matrix}\right.$.
\end{proof}

Given the form of Theorem~\ref{thm:poweroot}, we highlight a few special cases. If $p \rightarrow 0$, then $q \rightarrow 0$. If $p \rightarrow -\infty$, then $q \rightarrow 1$. If $p \rightarrow 1$, then $q \rightarrow -\infty$. This is consistent with what we know about constant reserve, sum, and product functions. The value function for the product invariant is also a product, which is shown by $p, q$ both converging to 0 (see Lemma~\ref{eq:special}). The value function for constant sum is a minimum function, which is captured by $q$ converging to 1 when $p$ converges to $-\infty$ (again, recall Lemma~\ref{eq:special}).  

\begin{figure}[h!]
\centering
\includegraphics[width=0.7\linewidth]{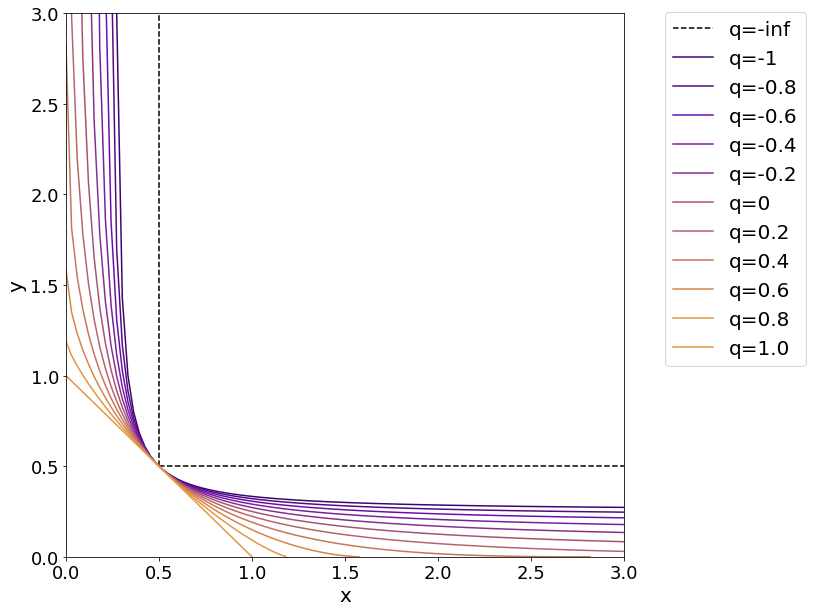}
\caption{Examples of power root that interpolate between $q=-1$ to $q=1$ (constant sum), including $q = 0$ (constant product). We plot the constant reserve in the dotted line.}
\label{fig:invariant}
\end{figure}

Figure~\ref{fig:invariant} visualizes the trading curve for a range of powers $q$. We include $q=0$ (constant product), $q=1$ (constant sum), and $q=-\infty$ (constant reserve, shown in dotted line) for comparisons. We observe that varying $q \in (-\infty, 1]$ predictably controls the curvature of the trading function between two extreme functions.

\textbf{Harmonic Mean Result }
In the case $q \rightarrow -1$, we recover the constant harmonic mean invariant.
By Theorem~\ref{thm:consistent}, we know $\psi_{\text{har}}(x, y)$ is concave.
Further, by Theorem~\ref{thm:poweroot}, we have that its value function is $V_{\text{har}}(a, b) = (a^{\frac{1}{2}}+b^{\frac{1}{2}})^2$ as $p = \frac{1}{2}$ when $q = -1$. This function has higher curvature than constant product.

\subsection{Marginal Price Function}
\label{sec:proot_price}

Given the constant power root trading function, we can derive an expression for the marginal price as a function of $p$. The ``marginal'' price represents the price of one token in terms of the other assuming an infinitesimally small trade. 

Begin with $\left((x - \Delta x)^q + (y + \Delta y)^q\right)^{\frac{1}{q}} = k$. Then $\Delta y = (k^q - (x - \Delta x)^q)^{\frac{1}{q}} - y$. Dividing by $\Delta x$, we obtain $\frac{\Delta y}{\Delta x} = \frac{(k^q - (x - \Delta x)^q)^{\frac{1}{q}} - y}{\Delta x}$. Derive the marginal price using L'Hopital's rule:
\begin{align*}
    \small
    M_{\text{pow}}(q) &= \lim_{\Delta x \rightarrow 0}\frac{\Delta y}{\Delta x} = \lim_{\Delta x \rightarrow 0} \left(\frac{\frac{d \Delta y}{d \Delta x}}{\frac{d \Delta x}{d \Delta x}} \right) = \lim_{\Delta x \rightarrow 0} \left( \frac{d}{d\Delta x} \left((k^q - (x - \Delta x)^q)^{\frac{1}{q}} - y\right)\right) \\
    &= \lim_{\Delta x \rightarrow 0} \left( \frac{1}{q}(k^q - (x - \Delta x)^q)^{\frac{1}{q} - 1} \cdot -q(x - \Delta x)^{q-1} \cdot -1 \right)\\
    &= (k^q - x^q)^{\frac{1}{q} - 1} \cdot x^{q-1} = (y^q)^{\frac{1}{q}-1} \cdot x^{q-1}
\end{align*}
Since $\frac{1}{q}-1 = \frac{1-q}{q}$, we obtain the following expression
\begin{equation}
    M_{\text{pow}}(q) = x^{q-1}y^{1-q}
    \label{eq:proot_price}
\end{equation}
A quick inspection of Equation~\ref{eq:proot_price} and Lemma~\ref{lem:price} shows that for $q = 1$, we have $M_{\text{pow}}(1) = x^{0}y^{0} = 1 = M_{\text{sum}}$, and $M_{\text{pow}}(0) = x^{-1}y^{1} = \frac{y}{x} = M_{\text{prod}}$. That is, the marginal price for the constant sum and product invariants are again, special cases of the constant power root.

Figure~\ref{fig:marginalprice} plots the marginal price as a function of the reserve ratios between $x$ and $y$. As $q \rightarrow 1$, we see the marginal price converge to a constant function at 1. As $q \rightarrow -\infty$, it approaches an exponentially increasing function. When $q = 0$, the price is linear.

\subsection{Impermanent Loss Function}
\label{sec:imploss}

Using the result in Section~\ref{sec:proot_price}, we show the following theorem:

\begin{theorem}
Let $\alpha > 0$ be a change in marginal price $M = M_{\text{pow}}(q)$. The impermanent loss for the constant power root trading function is
\begin{equation}
    I_{\textup{proot}}(q) = \left(\frac{1+M^{\frac{q}{1-q}}}{1 + (\alpha M)^{\frac{q}{1-q}}}\right)^{\frac{1}{q}}\left(\frac{\alpha M + (\alpha M)^{\frac{1}{1-q}}}{\alpha M +  M^{\frac{1}{1-q}}}\right) - 1.
\end{equation}
\label{eq:proot_imploss}
\end{theorem}
\begin{proof}
\small
First, we wish to represent $x$ and $y$ in terms of the constant $k$ and the marginal price $M$ only. Recall from Equation~\ref{eq:proot_price} that $M = x^{q-1}y^{1-q}$. So $y = M^{\frac{1}{1-q}}x$. Since $(x^q + y^q)^{\frac{1}{q}} = k$, we compute $x = \frac{k}{(1 + M^{\frac{q}{1-q}})^{\frac{1}{q}}}$, and thus $y = \frac{kM^{\frac{1}{1-q}}}{(1 + M^{\frac{q}{1-q}})^{\frac{1}{q}}}$.

Second, we can compute value $U(x, y, M) = M x + y$. Note $M' = \alpha M$. Let $x', y'$ be the token reserves at the new marginal price $M'$. Then $U(x', y', M') = M'\left(\frac{k}{(1 + (M')^{\frac{q}{1-q}})^{\frac{1}{q}}}\right) + \left(\frac{k(M')^{\frac{1}{1-q}}}{(1 + (M')^{\frac{q}{1-q}})^{\frac{1}{q}}}\right) = \frac{\alpha kM + k(\alpha M)^{\frac{1}{1-q}}}{(1 + (\alpha M)^{\frac{q}{1-q}})^{\frac{1}{q}}}$. Similarly, compute $U(x, y, M') = \frac{\alpha kM + k M^{\frac{1}{1-q}}}{(1 + M^{\frac{q}{1-q}})^{\frac{1}{q}}}$.

Now, use the definition that $I = \frac{U(x', y', M') - U(x, y, M')}{U(x, y, M')}$. To be pedantic, we obtain \[I = \left(\frac{\left(\frac{1 + M^{\frac{q}{1-q}}}{1 + (\alpha M)^{\frac{q}{1-q}}}\right)^{\frac{1}{q}}(\alpha M + (\alpha M)^{\frac{1}{1-q}}) - (\alpha M + M^{\frac{1}{1-q}})}{\alpha M + M^{\frac{1}{1-q}}}\right).\] Divide by the denominator to find the result.
\end{proof}

We aim to show special cases of this impermanent loss. As these require nontrivial work, we do so in the following corollaries.
\begin{corollary}
$\lim_{q \rightarrow 0} I_{\textup{proot}}(q) = I_{\textup{prod}}$. 
\label{coro:one}
\end{corollary}
\begin{proof}
We compute the limit separately for the three pieces of $I_{\textup{proot}}(q)$.

\begin{enumerate}
    \item $\lim_{q\rightarrow 0} \left( \frac{1 + M^{\frac{q}{1-q}}}{1 + (\alpha M)^{\frac{q}{1-q}}} \right)^{\frac{1}{q}} = \exp  \lim_{q \rightarrow 0}  \left(\frac{\log \left(\frac{1 + M^{\frac{q}{1-q}}}{1 + (\alpha M)^{\frac{q}{1-q}}}\right)}{q} \right)   = \left(\frac{M}{\alpha M}\right)^{\frac{1}{2}} = \alpha^{-\frac{1}{2}}$.
    \item $\lim_{q\rightarrow 0}\left(\frac{\alpha M + (\alpha M)^{\frac{1}{1-q}}}{\alpha M +  M^{\frac{1}{1-q}}}\right) =  \frac{2\alpha M}{(\alpha+1)M} = \frac{2\alpha}{\alpha + 1}$.
\end{enumerate}

Altogether, the limit evaluates to $\frac{1}{\sqrt{\alpha}} \cdot \frac{2\alpha}{\alpha + 1} - 1 = \frac{2\sqrt{\alpha}}{\alpha+1} - 1$. Recall Lemma~\ref{lem:og:imp} to see this is equivalent to $I_{\textup{prod}}$.
\end{proof}

\begin{corollary}
$\lim_{q \rightarrow 1} I_{\textup{proot}}(q) = I_{\textup{sum}}$. 
\label{coro:two}
\end{corollary}
\begin{proof}
\small
Consider only the first term in $I_{\textup{proot}}$ and apply exp-log to get separate limits. Let $J(q) = I_{\text{pow}}(q) + 1 = \left(\frac{1+M^{\frac{q}{1-q}}}{1 + (\alpha M)^{\frac{q}{1-q}}}\right)^{\frac{1}{q}}\left(\frac{\alpha M + (\alpha M)^{\frac{1}{1-q}}}{\alpha M  +  M^{\frac{1}{1-q}}}\right)$. Then, 
\[\lim_{q \rightarrow 1} J(q) = \exp \left( \log \lim_{q \rightarrow 1} A  + \log \lim_{q \rightarrow 1} B  \right).\] Focusing on the expression inside the $\log$, we will study each of these in turn. First,
\begin{align*}
    \log \lim_{q\rightarrow 1} A &= \lim_{q\rightarrow 1} \log \left( \frac{1 + M^{\frac{q}{1-q}}}{1 + (\alpha M)^{\frac{q}{1-q}}} \right)^{\frac{1}{q}} = \lim_{q \rightarrow 1} \frac{1}{q} \cdot \log \left( \lim_{q \rightarrow 1} \frac{1 + M^{\frac{q}{1-q}}}{1 + (\alpha M)^{\frac{q}{1-q}}}  \right) \\
    &= \log \left( \frac{\frac{M^{\frac{-q}{q-1}}\log M}{(q-1)^2}}{\frac{(\alpha M)^{\frac{-q}{q-1}}\log (\alpha M)}{(q-1)^2}} \right) = \log \left( \lim_{q \rightarrow 1} \alpha^{\frac{-q}{q-1}} \cdot \frac{\log M }{\log (\alpha M)} \right) \\
    &= \lim_{q \rightarrow 1} \left( \frac{q}{q-1} \right) \log \alpha + \log \left(\frac{\log M}{\log (\alpha M)}\right)
\end{align*}
where the swaps between the limit and log follows by the continuity of log. 

This limit is undefined but we can evaluate the left and right limits. If $q \rightarrow 1^+$, then $\left(\frac{q}{q-1}\right) \rightarrow \infty$, leaving $\log \lim_{q\rightarrow 1^+} A = \left\{\begin{matrix}
 -\infty & 0 < \alpha < 1 \\
 0 & \alpha = 1  \\
 \infty & \alpha > 1
\end{matrix}\right.$. If $q \rightarrow 1^-$, we have $\left(\frac{q}{q-1}\right) \rightarrow -\infty$, leaving $\log \lim_{q\rightarrow 1^-} A = \left\{\begin{matrix}
 \infty & 0 < \alpha < 1 \\
 0 & \alpha = 1  \\
 -\infty & \alpha > 1
\end{matrix}\right.$. 

Next, we can evaluate the second term:
\begin{align*}
    \log \lim_{q\rightarrow 1} B &= \log \lim_{q \rightarrow 1}\left(\frac{\alpha M + (\alpha M)^{\frac{1}{1-q}}}{\alpha M + M^{\frac{1}{1-q}}} \right) = \log \lim_{q \rightarrow 1} \left( \frac{\frac{(\alpha M)^{\frac{-1}{q-1}} \log (\alpha M) }{(q-1)^2}}{\frac{M^{\frac{-1}{q-1}} \log M}{(q-1)^2}}\right) \\ 
    &= \log \lim_{q\rightarrow 1}\left( \ \frac{(\alpha M)^{\frac{-1}{q-1}} \log (\alpha M)}{M^{\frac{-1}{q-1}} \log M}\right) = \log \lim_{q \rightarrow 1} \left( \alpha^{\frac{-1}{q-1}} \cdot \frac{\log (\alpha M)}{\log M} \right)
\end{align*}
where the second equality follows by L'Hopital's rule.

Again, this limit doesn't exist but we can evaluate the left and right limits. If $q \rightarrow 1^+$, then $\left(\frac{-1}{q-1}\right) \rightarrow -\infty$, leaving $\log \lim_{q\rightarrow 1^+} A = \left\{\begin{matrix}
 \infty & 0 < \alpha < 1 \\
 0 & \alpha = 1  \\
 -\infty & \alpha > 1
\end{matrix}\right.$. If $q \rightarrow 1^-$, we have $\left(\frac{-1}{q-1}\right) \rightarrow \infty$, leaving $\log \lim_{q\rightarrow 1^-} A = \left\{\begin{matrix}
 -\infty & 0 < \alpha < 1 \\
 0 & \alpha = 1  \\
 \infty & \alpha > 1 \\
\end{matrix}\right.$.

Now, we notice that $\log \lim A + \log \lim B = 0$ for all cases of $\alpha$ and $q$ coming from left or right. Note we cannot take the limit of $A + B$ as each is not defined.

Thus, we have $\lim_{q \rightarrow 1} J(q) = \exp 0 = 1$. So $\lim_{q \rightarrow 1} I_{\textup{proot}}(q) = 0$, which is $I_{\textup{sum}}$.
\end{proof}
Corollaries~\ref{coro:one} and \ref{coro:two} show that like the marginal price, the constant sum and product impermanent loss functions are special cases of the constant power root. 
\begin{figure}[h!]
\centering
\includegraphics[width=\textwidth]{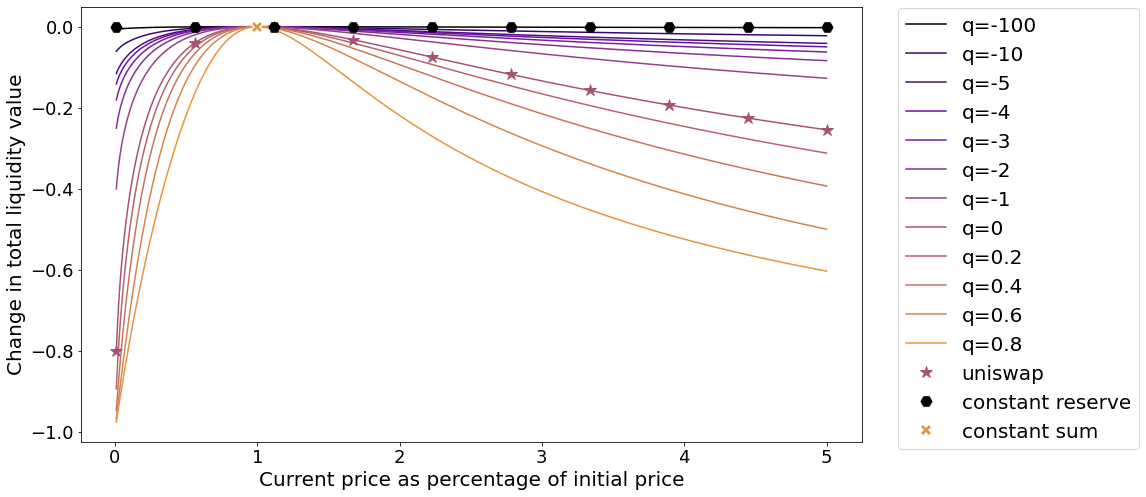}
\caption{Impermanent loss as a function of $-\infty < q \leq 1$ with comparisons to constant sum, constant reserve, and constant product market makers (i.e. Uniswap).}
\label{fig:imploss}
\end{figure}

We plot the impermanent loss as a function of $q$ in Figure~\ref{fig:imploss}, where the constant product (i.e. Uniswap) is shown in purple stars, constant reserve is shown in black hexagons, and the constant sum is shown with an orange cross (note this is only a single point as marginal price does not change for the constant sum invariant). Figure~\ref{fig:imploss} shows a gradient of loss functions as $q$ varies. In particular, we see increased impermanent loss as $q$ approaches 1 and reduced impermanent loss as $q$ approaches $-\infty$. This is aligned with prior work \cite{angeris2020does} that suggests trading functions with flatter curvature have higher loss.

\subsection{Price Impact}

Next, we compute the price impact function \cite{angeris2020does}, or the change to marginal price given a trade of size $\Delta x$, as $\frac{d \Delta y}{d \Delta x}$, the derivative of the price function. Equation~\ref{eq:priceimpact} shows the price impact function for power roots.
\begin{equation}
    \frac{d \Delta y}{d \Delta x} = (x^q + y^q - (x - \Delta x)^q)^{\frac{1-q}{q}} \cdot (x - \Delta x)^{q-1}
    \label{eq:priceimpact}
\end{equation}
Refer to \cite{angeris2020does} for derivations for the constant sum and product functions. It is straightforward to show that those are special cases of Equation~\ref{eq:priceimpact} by $q = 0$ and $q = 1$.
Figure~\ref{fig:priceimpact} shows the price impact for constant power root market makers of varying $q$. When $q \rightarrow 1$, the price impact function approaches 0 slope, which makes sense given that the marginal price is constant when $q = 1$. When $q \rightarrow -\infty$, the price impact function has a faster increasing slope, representing higher levels of price sensitivity.

Practically, we note that price impact function is closely related to slippage as a constant function market maker with high price impact will likely experience higher slippage rates. As a base example, when $q = 1$, there is  zero slippage due to the constant marginal price, or equivalently, its price impact function having zero slope.

\begin{figure}[h!]
\centering
\begin{subfigure}[b]{.45\textwidth}
    \centering
    \includegraphics[width=\linewidth]{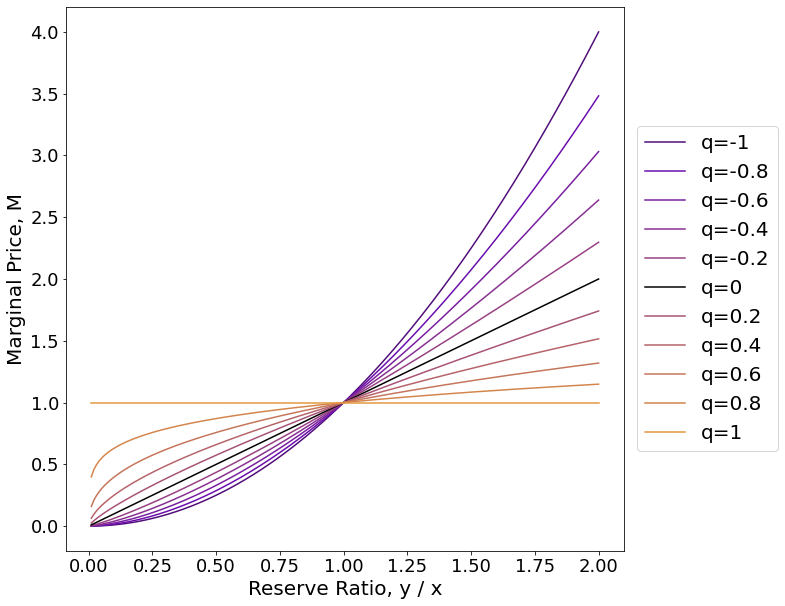}
    \caption{Marginal price as a function of $q$. On the x-axis, we plot a ratio of $\frac{y}{x}$ with $x = 1$.}
    \label{fig:marginalprice}
\end{subfigure}
\hfill
\begin{subfigure}[b]{.45\textwidth}
    \centering
    \includegraphics[width=\linewidth]{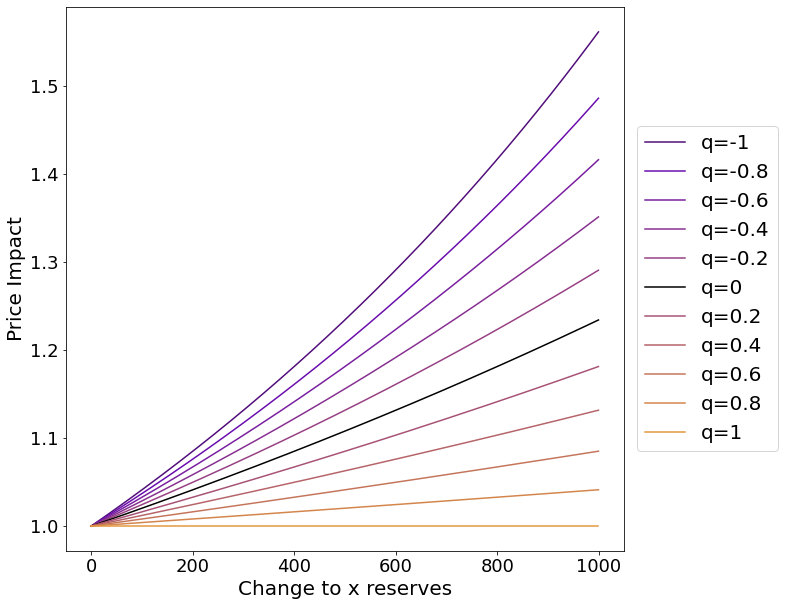}
    \caption{Impact of a trade on price as a function of the trade size ($\Delta x \in [0, 9000]$, $x, y = 10000$).}
    \label{fig:priceimpact}
\end{subfigure}
\caption{Marginal price is shown on the left, defined as the price in the limit of an infinitesimally small trade. Price impact is shown on the right, the sensitivity of price in response to trade size. In both figures, we see a continuum between $q=-1$ and $q=1$.}
\end{figure}

In summary, there is a trade-off between impermanent loss and price impact, which is perhaps not surprising. For large negative powers, the invariant has low loss for LPs but high price slippage for traders. For powers close to one, the opposite is true. 

\subsection{Value Derivatives}

Since impermanent loss can be difficult to interpret, we also compute the value $U(x,y,M)$ as well as derivatives with respect to price, such as delta $\Delta = \frac{\delta U}{\delta M}$ and gamma $\Gamma = \frac{\delta^2 U}{\delta M^2}$. Continuing from Section~\ref{sec:imploss}, we can simplify the expression for value using only the liquidity amount $k$ and the marginal price $M$, removing reserves from the equation:
\begin{equation}
    U(M,k) = Mk\left(1 + M^{\frac{q}{1-q}}\right)^{\frac{q-1}{q}}
\end{equation}
Taking first and second derivatives, we arrive at closed form expressions for the greeks:
\begin{equation}
    \Delta(M,k) = k\left(M^{\frac{q}{q-1}} + 1\right)^{-\frac{1}{q}},\qquad \Gamma(M, k) = \left(\frac{-1}{1-q}\right)M^{\frac{q}{1-q}-1}\left(M^{\frac{q}{1-q}} + 1\right)^{\frac{-1}{q}-1}
\end{equation}
In the following figure, we visualize the greeks for a variety of powers $q$.

\begin{figure}[h!]
\centering
\begin{subfigure}[b]{.49\textwidth}
    \centering
    \includegraphics[width=\linewidth]{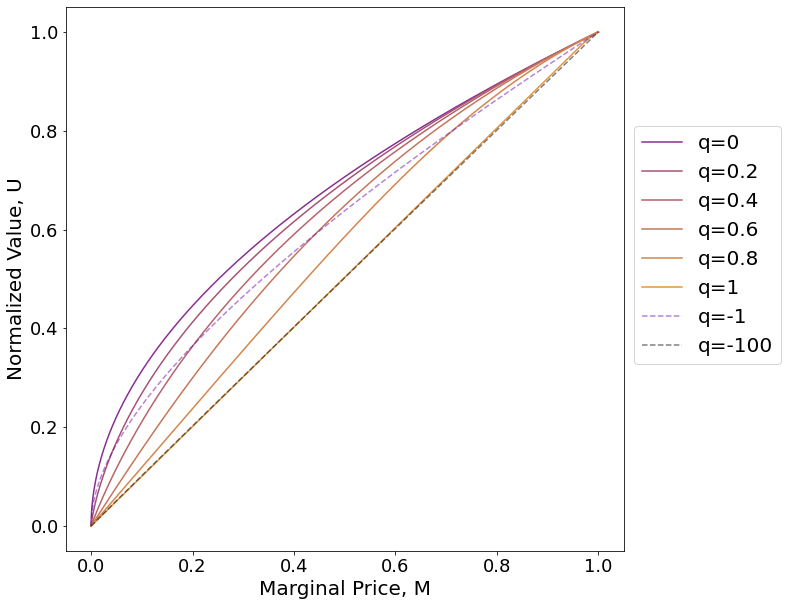}
    \caption{Value $U$}
    \label{fig:value}
\end{subfigure}
\hfill
\begin{subfigure}[b]{.49\textwidth}
    \centering
    \includegraphics[width=\linewidth]{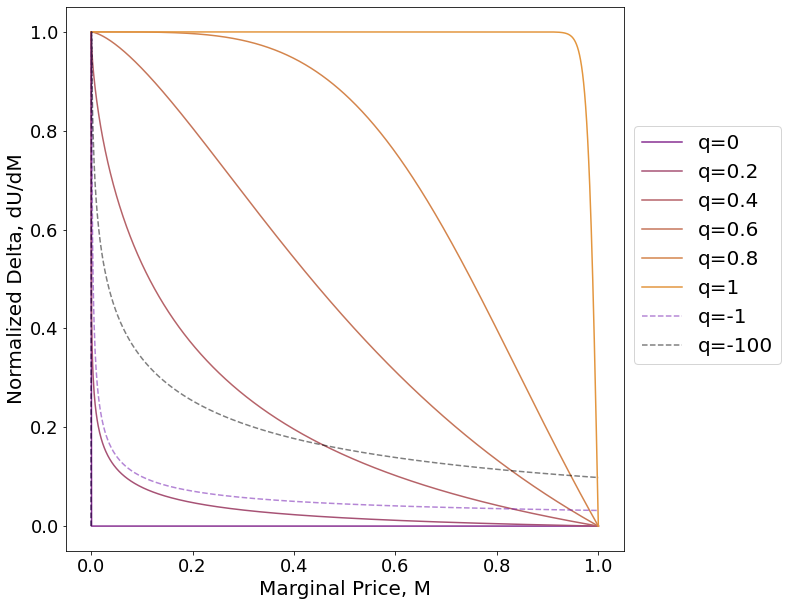}
    \caption{Delta $\Delta$}
    \label{fig:delta}
\end{subfigure}
\begin{subfigure}[b]{.5\textwidth}
    \centering
    \includegraphics[width=\linewidth]{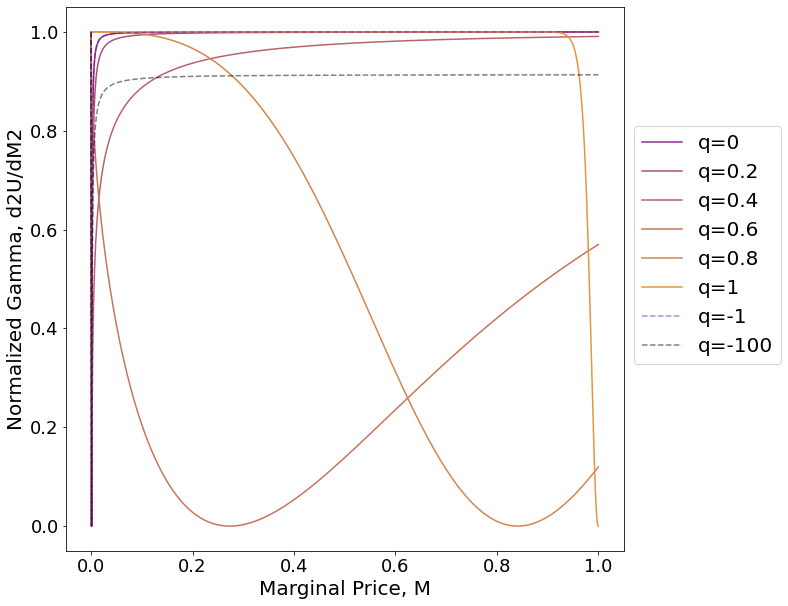}
    \caption{Gamma $\Gamma$}
    \label{fig:gamma}
\end{subfigure}
\caption{Visualizing value, delta, and gamma for a range of powers. In particular, we compare $q = 0$ (constant product, Uniswap) to $q = 1$ (constant sum) to $q = -1$ (constant harmonic mean) to $q = -100$ (a proxy for $q = -\infty$, also called constant reserve or HODL).}
\label{fig:greeks}
\end{figure}

In Figure~\ref{fig:greeks}, we vary the marginal price $M$ between 0 and 1 and study the effect on value, delta, and gamma. For each subplot, we normalize the y axis to be between 0 and 1 for visual comparison. As such, the reader should focus on the shape rather than the magnitudes. We make a few observations. First, in Figure~\ref{fig:greeks}(a), for positive powers, the larger the power, the less curvature in the value function. Interestingly, for larger negative powers, we observe a similar behavior, as $q = -100$ and $q = 1$ have similar shape. Recalling payoff functions from Figure~\ref{fig:invariant}, we note that Figure~\ref{fig:greeks}(a) shows the same plot in lower dimensions.

Second, we observe more diverse delta functions as $q$ spans between 0 and 1 in Figure~\ref{fig:greeks}(b). For smaller positive powers, we observe delta to increase exponentially as price goes to zero, whereas for larger positive powers, we observe delta to increase exponentially as price grows large. For powers close to $q=0.5$, we observe near linear change in delta. Third, in Figure~\ref{fig:greeks}(c), we observe that subtle differences in value are exaggerated as gamma functions across powers differ in behavior. This large range in gamma functions may suggest power roots to play a role in structured products to achieve a desired gamma. We note that delta and gamma for negatives powers  closely mimic those of constant product.

\subsection{Reserve Depletion}

A benefit of the Uniswap trading function ($q=0$) is a guarantee of no reserve depletion. For any value of token $x > 0$ and any constant $k$, we have a $y = \frac{k}{x} > 0$. In practice, this implies no trade can conduct a trade large enough to deplete resources. This property, however, does not hold for the sum invariant ($q=1$), where the analogous expression $y = k - x$ suggests that $x \leq k$, otherwise $y \leq 0$.
% , a non-sensible result. This is largely the reason why constant sum market makers are not used in practice as it is trivial for an arbitrageur to trade $x > k$. 

The power root invariant with $q \leq 0$ guarantees no reserve depletion. However for powers $0 < q \leq 1$ it, like the sum, has the same point of reserve depletion, namely $x = k$. One can see this from the equation $y = (k^q - x^q)^{\frac{1}{q}}$, which implies $k^q \geq x^q$ is required. As $q \rightarrow 0$, this constraint is removed, and power root collapses to product. However, for $q$ closer to 0, we observe more incentives to not deplete reserves due to price increases, though it is theoretically possible. For example, for $q = \epsilon$ close to 0, the marginal price is $\frac{y^{1 - \epsilon}}{x^{1-\epsilon}} \approx \frac{y}{x}$. If $x$ is near depleted (and $y$ near $k$), the price will be tending towards $\infty$.  

\begin{figure}[h!]
\centering
\includegraphics[width=0.8\textwidth]{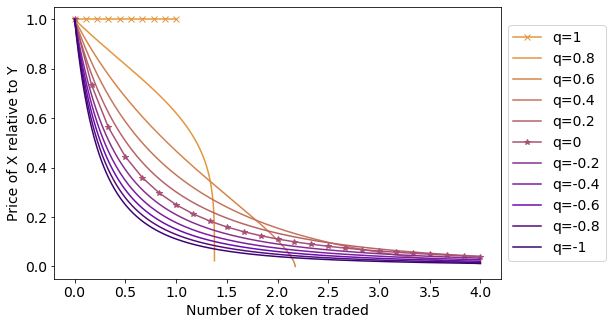}
\caption{The price of token $x$ relative to token $y$ as a function of the trade size.}
\label{fig:relative}
\end{figure}

Figure~\ref{fig:relative} visualizes the relative price between two tokens as a function of trade size in token $x$. Constant sum and product are shown with special markers. We make a few observations. First, the relative price decreases as trade size grows, which is expected given the concavity of the trading functions. Second, for $0 < q < 1$, the plotted lines intersect the x-axis, representing points of depleted reserves (for $q = 1$, we have the point of depletion at $x = 1$). For $q \leq 0$, the plotted lines never intersect the x-axis. Third, for $0 < q < 1$, there are segments where the price is better than constant product, which suggests better conditions for traders assuming small trades. This is not the case for large trades. Clipper \cite{othman2021} formalized this approach by only allowing small trades. For $q < 0$, prices are bounded above by constant product ($q=0$), meaning traders face a higher cost.

\section{Summary}
We proposed a new constant function market maker using the power root equation $\psi_{\text{pow}}(\mathbf{x}) = \left(\sum_i x_i^q\right)^{\frac{1}{q}}$ with the corresponding payoff (value) function for liquidity providers as $V_{\text{pow}}(\mathbf{c}) = \left(\sum_i c_i^p\right)^{\frac{1}{p}}$ with $q = \frac{p}{p-1}$. Inspired by economic production functions, we showed this choice to be a generalization of the constant reserve, sum, and product. The constant power root function interpolates between the constant reserve, sum, and product functions in terms of marginal price, price impact, and impermanent loss. In our analysis, we observe a trade-off between low slippage and impermanent loss, and susceptibility to reserve depletion. 
% Finally, we proposed the dynamic power root market maker, which aims to vary the power based on the reserve ratio and liquidity. Given the simulation results, we are optimistic this new trading function may result in more capital efficient markets.

\section*{Acknowledgement}

A special thanks must go to Guillaume Lambert, Alok Vasudev, Theo Diamandis, Dan Robinson, and Joey Santoro for their helpful feedback and discussion.

\bibliographystyle{unsrt}
\bibliography{main}

\end{document}